\documentclass[onecolumn]{IEEEtran}
\usepackage{caption}

\usepackage{mathtools}
\usepackage{algorithm}
\usepackage{algpseudocode}
\usepackage{amsmath}
\usepackage{amssymb}
\usepackage{booktabs}
\usepackage{array}
\usepackage{graphicx}
\newcolumntype{C}{>{$}c<{$}} % math-centered column
\usepackage{fullpage}
\usepackage{amsthm}
\usepackage{color}
\usepackage{balance}
\newcommand{\newreptheorem}[2]{%
\newenvironment{rep#1}[1]{%
 \def\rep@title{#2 \ref{##1}}%
 \begin{rep@thm}}%
 {\end{rep@thm}}}
\newreptheorem{theorem}{Theorem}
\newtheorem*{rep@thm}{\rep@title}
\usepackage{pgfplots}
\pgfplotsset{compat=1.18} 

\usepackage[utf8]{inputenc} 
\usepackage[T1]{fontenc}
\usepackage{url}
\usepackage{ifthen}
\usepackage{nicefrac, xfrac}
\usepackage[hidelinks]{hyperref}
\usepackage{cleveref}
\usepackage{epsfig, cite}
\usepackage{amsfonts}
\usepackage{soul}
\usepackage{latexsym}
\usepackage{colortbl}
\usepackage{comment}
\usepackage[dvipsnames]{xcolor}
\usepackage{caption}
\usepackage{subcaption}

\newcommand{\Fq}{\mathbb{F}_q}

\newtheorem{thm}{Theorem}[section]
\newtheorem{lem}[thm]{Lemma}
\newtheorem{prop}[thm]{Proposition}
\newtheorem{claim}[thm]{Claim}
\newtheorem{cor}{Corollary}

\newtheorem{defn}{Definition}[section]
\newtheorem{exmp}{Example}[section]
\newtheorem{rem}{Remark}
\newtheorem{const}{Construction}[section]

\newcommand{\roundDown}[1]{\ensuremath{\lfloor#1\rfloor}}

\newcommand{\bfa}{\ensuremath{\boldsymbol{a}}}
\newcommand{\bfb}{\ensuremath{\boldsymbol{b}}}
\newcommand{\bfc}{\boldsymbol{c}}

\newcommand{\bfh}{\ensuremath{\boldsymbol{h}}}

\newcommand{\bfs}{\ensuremath{\boldsymbol{s}}}

\newcommand{\bfv}{\ensuremath{\boldsymbol{v}}}

\newcommand{\bfx}{\ensuremath{\boldsymbol{x}}}
\newcommand{\bfy}{\ensuremath{\boldsymbol{y}}}

\newcommand{\bfzero}{\boldsymbol{0}}
\newcommand{\floor}[1]{\lfloor #1 \rfloor}
\newcommand{\LCS}{\textup{LCS}}

\newcommand{\cR}{\mathcal{R}}
\algrenewcommand\algorithmicrequire{\textbf{Input:}}
\algrenewcommand\algorithmicensure{\textbf{Output:}}

\begin{document}

\title{Improved Constructions of Linear Codes for Insertions and Deletions} 
\author{%
\IEEEauthorblockN{\textbf{Roee~Gross}\IEEEauthorrefmark{1}, \textbf{Roni~Con}\IEEEauthorrefmark{1} and {\textbf{Eitan~Yaakobi}\IEEEauthorrefmark{1}}}\\
   \IEEEauthorblockA{\IEEEauthorrefmark{1}%
    Department of Computer Science, %\\
    Technion---Israel Institute of Technology, Haifa 3200003, Israel    
    Email: \{roee.g, roni.con, yaakobi\}@cs.technion.ac.il}

\thanks{%
This work was supported in part by the Israel Science Foundation (ISF) grant 2462/24 and by the European Union (EIC, DiDAX 101115134). Views and opinions expressed are however those of the authors only and do not necessarily reflect those of the European Union or the European Research Council Executive Agency. Neither the European Union nor the granting authority can be held responsible for them.}

}

\maketitle

\begin{abstract}
In this work, we study linear error-correcting codes against adversarial insertion-deletion (indel) errors. While most constructions for the indel model are nonlinear, linear codes offer compact representations, efficient encoding, and decoding algorithms, making them highly desirable. A key challenge in this area is achieving rates close to the half-Singleton bound for efficient linear codes over finite fields. We improve upon previous results by constructing explicit codes over \(\mathbb{F}_{q^2}\), linear over \(\mathbb{F}_q\), with rate \(1/2 - \delta - \varepsilon\) that can efficiently correct a \(\delta\)-fraction of indel errors, where \(q = O(\varepsilon^{-4})\). Additionally, we construct fully linear codes over \(\mathbb{F}_q\) with rate \(1/2 - 2\sqrt{\delta} - \varepsilon\) that can also efficiently correct \(\delta\)-fraction of indels. These results significantly advance the study of linear codes for the indel model, bringing them closer to the theoretical half-Singleton bound. We also generalize the half-Singleton bound, for every code \(C \subseteq \mathbb{F}^n\) linear over \(\mathbb{E} \subset \mathbb{F}\) a subfield of $\mathbb{F}$, such that \(C\) has the ability to correct \(\delta\)-fraction of indels, the rate is bounded by $(1-\delta)/2$.
\end{abstract}

\newpage
\tableofcontents
\newpage

% Keywords
% \begin{IEEEkeywords}
% %Synchronization errors, Indel correction, Linear codes.
% \end{IEEEkeywords}
\section{Introduction}
Error-correcting codes are a fundamental tool in information theory and theoretical computer science, enabling reliable communication over noisy channels. Traditionally, the study of error correction has focused on two primary corruption models, substitutions and erasures. In these models, each symbol in a transmitted word can be replaced with another symbol (substitution) or marked as unknown (erasure). These classical frameworks, introduced by the seminal works of Shannon~\cite{Shannon} and Hamming~\cite{Hamming}, have been extensively studied, leading to efficient constructions of codes that are both encodable and decodable.

However, beyond substitution and erasure errors, another type of corruption, synchronization errors, poses unique challenges. Synchronization errors directly affect the length of the transmitted word, making them fundamentally different from substitution and erasure errors. The most widely studied framework for synchronization errors is the \emph{insertion-deletion (indel) model}. An insertion adds a symbol between existing symbols, while a deletion removes a symbol entirely. 

This natural theoretical model and possible applications across many fields, including the emerging DNA-based storage systems, has led many researchers in the information theory and computer science communities to study codes for these errors.
And indeed, there has been significant progress in recent years on understanding this model of indel errors (both on limitation and constructing efficient codes). Still, our comprehension of this model lags far behind our understanding of codes that correct erasures and substitution errors (we refer the reader to the following excellent surveys \cite{mitzenmacher2009survey,mercier2010survey,cheraghchi2020overview,haeupler-survey2021synchronization}).

It might come as a surprise that
most of the works constructing codes for the indel model are not linear codes. 
Indeed, linear codes are the dominant class of codes in the substitutions and erasures error models, where some notable examples include Reed--Solomon, Reed--Muller, Polar code, algebraic-geometry codes, and many more. 
The reason for the absence of linear codes in this model 
appears in \cite{abdel2007linear} where it was shown that \emph{any} linear code correcting a single indel error must have rate at most $1/2$. This shows that linear codes are provably worse than non-linear as non-linear codes correcting $1$ indel error can have rate $1- o(1)$. 

However, linear codes have many strong advantages over nonlinear codes. They have compact representations (generating/parity check matrices), they are efficiently encodable and in many cases, have an efficient decoding algorithm. 
Thus, studying them in the indel model was a subject of many recent works \cite{Cheng_2021,liu2022bounds,con2023reed,ji2023strict,cheng2023linear,con2023optimal,liu2024optimal,con2024random,xie2024new,li2025linear}.
One of the main questions regarding linear codes against indels is to design explicit and efficient codes over constant size alphabets that achieve the \emph{half-Singleton} bound. That is, the ultimate goal is to explicitly construct linear codes of $R = \frac{1}{2}(1 - \delta) - \varepsilon$ over fields of size $\textup{poly}(1/\varepsilon)$ that can correct efficiently from $\delta$-fraction of indel errors. 

In this paper, we improve the results of \cite{CST22}. Specifically, we construct codes over $\mathbb{F}_{q^2}$ where $q = \Theta(\varepsilon^{-4})$ that are linear over $\Fq$, can efficiently correct from $\delta$-fraction of indel errors, and have rate $1/2 - \delta - \varepsilon$.
Then, we show how to construct \emph{linear} codes over $\Fq$ of rate $1/2 - 2\sqrt{\delta} - \varepsilon$ that can efficiently correct from $\delta$ fraction of indel errors where $q = \Theta(\varepsilon^{-4})$. 
Finally, we show that the half-Singleton bound for linear codes holds also for codes that are linear only over a subfield.

\subsection{Previous works}
The model of correcting from indel errors was first introduced by Levenshtein \cite{levenshtein1966binary} and who showed that a code correcting $t$ deletions can, in fact, correct any combination of $t$ indel errors. Levenshtein also showed that the codes (that were originally designed to correct a single asymmetric error) by Varshamov and Tenengolts  \cite{varshamov1965codes} are asymptotically optimal codes (in terms of required redundancy) for correcting a single indel error. However, even to this date, the question of what is the required redundancy to correct a constant number of indel errors is not known (see the work of Alon et al. \cite{alon2023logarithmically} and references within). Also, there are many ingenious explicit constructions of codes with low redundancy that correct a constant number of deletions \cite{gabrys2018codes,sima2019two,brakensiek2017efficient,sima2020optimal,sima-q2020optimal,guruswami2021explicit,liu2024explicit}, just to name a few.

Our interest in this paper is in the regime where the number of indel errors is a \emph{constant fraction} of the codeword length (rather than a constant number which does not depend on the codeword length). We start with mentioning the known nonlinear codes and then focus on linear codes.

\textbf{Explicit nonlinear binary codes correcting a constant fraction of indels.} To the best of our knowledge, the first construction of \emph{asymptotically good}\footnote{By asymptotically good, we mean that the rate of the code and the fraction of the indels it can correct are numbers bounded away from $0$.} binary codes correcting indel errors is due to Schulman and Zuckerman \cite{schulman2002asymptotically}. In \cite{haeupler2017synchronization}, Haeupler and Shahrasbi constructed codes over an alphabet of size $\exp(O(1/\varepsilon))$ capable of correcting efficiently $\delta$-fraction of indels and have rate $1 - \delta - \varepsilon$. 
Note that a code correcting $\delta$-fraction of indels must have rate at most $1 - \delta$ by a simple Singleton bound and thus, the codes of \cite{haeupler2017synchronization} achieve optimal rate-error-correction trade-off. As discussed above, this shows that linear codes are provably worse than nonlinear codes when recovering from indels. 
Later, it was shown in \cite{haeupler2018synchronization} that an alphabet of size $\exp(-\Omega(1/\varepsilon))$ is needed to achieve a code correcting $\delta$-fraction of indels with rate $1 - \delta -\varepsilon$.
For \emph{binary} codes correcting a constant fraction of indels, the state-of-the-art efficient constructions are due to Cheng et al. \cite{cheng2018deterministic} and Haeupler \cite{haeupler2019optimal} who, independently constructed binary codes of rate $1 - O(\delta\log^{2}(1/\delta))$ correcting efficiently $\delta$-fraction of indels.

\textbf{Linear codes correcting indel errors.}
As mentioned above, the first work that studied the performance of linear codes against indels was by
Abdel-Ghaffar, Ferreira, and Cheng \cite{abdel2007linear} who proved that any linear code correcting even $1$ indel, must have rate at most $1/2$. Then, Cheng, Guruswami,  Haeupler, and Li~\cite{Cheng_2021} extended this bound and showed that a linear code correcting $\delta$-fraction of indels, must have rate
\(
\mathcal{R} \leq \frac{1}{2}\left( 1 - \frac{q}{q-1}\cdot \delta \right)\leq \frac{1}{2}(1 - \delta)
\)
where the first bound is termed as \emph{half-Plotkin bound} and the second bound as \emph{half-Singleton bound}. 
More specific upper bounds on special families of linear codes correcting indel errors were given in \cite{chen2022coordinate,ji2023strict,xie2024new}.

We now turn our focus to efficient linear codes correcting indel errors. The first ones to provide asymptotically good linear codes efficiently correcting indels are \cite{Cheng_2021}.  Specifically, they constructed \emph{binary} linear codes of rate $\approx 2^{-80}$ correcting $\delta < 1/400$ indel errors.

Then, \cite{CST22} constructed  for any $\varepsilon >0$, a linear code over a field of size $q=\Theta(\varepsilon^{-4})$, correcting $\delta$ fraction of indel errors with
\begin{align*}
   % R_{\nicefrac{1}{2}}(\delta, q) &\ge \frac{1}{4}(1 - \delta) - \Theta(q^{-1/4})\\
    \mathcal{R} &\ge \frac{1}{8}(1 - 4\delta)- \varepsilon\;.
\end{align*}
They also constructed binary linear codes capable of correcting $\delta \leq 1/54$ fraction of deletions and achieve rate $(1-54 \delta)/1216$. We note that \cite{CST22} also considered the relaxation of codes which are linear over a subfield of the field. 
They showed that over $\mathbb{F}_{q^2}$, one can construct codes which are linear over $\Fq$, can correct (efficiently) $\delta$ fraction of insdel errors and have rate $(1-\delta)/4 - \varepsilon$.

Later, Cheng at al.~\cite{cheng2023linear} focused on the high rate and high noise regimes of linear codes correcting indel errors.
They constructed \emph{binary} linear codes of rate $\mathcal{R} = 1/2 - \varepsilon$ capable of correcting $\Theta(\varepsilon^{-3}/\log (1/\varepsilon))$ indel errors efficiently. 
Also, in the high noise regime, they provided constructions of linear codes correcting $1 - \varepsilon$ indels over alphabets of size $\textup{poly($\frac{1}{\varepsilon}$)}$ with rates $\Omega (\varepsilon^2)$ for inefficient decoding and with rate $\Omega(\varepsilon^4)$ with efficient decoding.

More recently, Li, Gabrys, and Farnoud \cite{li2025linear} explored list decoding from indels of linear codes \cite{li2025linear}. They showed a construction of linear codes of rate $1 - \varepsilon$ capable of correcting $\Omega(\varepsilon^4)$-fraction of indels. 
This result shows that the half-Singleton bound breaks when considering list decoding instead of unique decoding. We also remark here that there is a recent line of work that studied the performance of Reed--Solomon codes under indels \cite{con2023reed,con2023optimal,con2024random,liu2024optimal,beelen2025reed}. 

\subsection{Our Results}

\begin{figure}[t]
  \centering
  % First subfigure
  \begin{subfigure}[t]{0.48\linewidth}
    \centering
    \begin{tikzpicture}
      \begin{axis}[
        width=\linewidth,
        xlabel={$\,\delta$},
        ylabel={$R(\delta)$},
        xmin=0, xmax=1,
        ymin=0, ymax=0.55,
        grid=both,
        legend style={at={(0.5,0.8)},anchor=west,font=\footnotesize},
        legend cell align=left]
        \addplot[domain=0:1,samples=200,very thick] {0.5 - 0.5*x};
        \addlegendentry{Half-Singleton bound}
        \addplot[domain=0:1,samples=200,very thick,dotted] {0.25 - 0.25*x};
        \addlegendentry{CST22}
        \addplot[domain=0:0.5,samples=200,very thick,dashed] {1/2 - x};
        \addlegendentry{\Cref{thm:half-linear-code}}
      \end{axis}
    \end{tikzpicture}
    \caption{Half-linear case}
    \label{fig:half-linear-rate}
  \end{subfigure}%
  \hfill
  % Second subfigure
  \begin{subfigure}[t]{0.48\linewidth}
    \centering
    \begin{tikzpicture}
      \begin{axis}[
        width=\linewidth,
        xlabel={$\,\delta$},
        ylabel={$R(\delta)$},
        xmin=0, xmax=1,
        ymin=0, ymax=0.55,
        grid=both,
        legend style={at={(0.5,0.8)},anchor=west,font=\footnotesize},
        legend cell align=left]
        \addplot[domain=0:1,samples=200,very thick] {0.5 - 0.5*x};
        \addlegendentry{Half-Singleton bound}
        \addplot[domain=0:0.25,samples=200,very thick,dotted] {0.125 - 0.5*x};
        \addlegendentry{CST22}
        \addplot[domain=0:0.0625,samples=200,very thick,dashed] {0.5 - 2*sqrt(x)};
        \addlegendentry{\Cref{thm:full linear results}}
      \end{axis}
    \end{tikzpicture}
    \caption{Linear case}
    \label{fig:linear-rate}
  \end{subfigure}

  \caption{Code rate $R(\delta)$ vs.\ indel fraction $\delta$ for different cases.}
  \label{fig:combined-rate}
\end{figure}
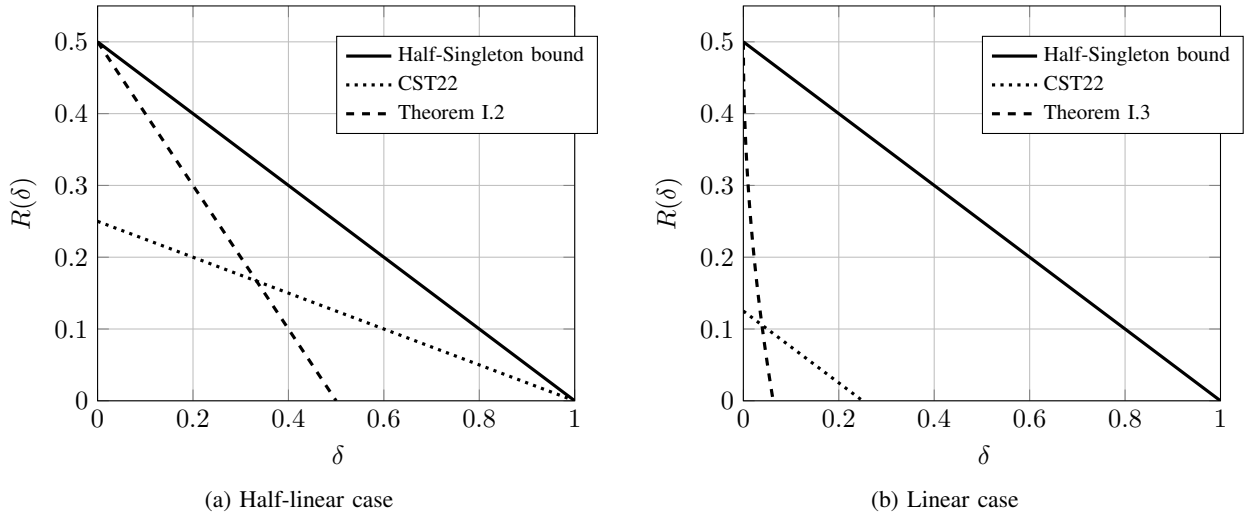

Our first result shows that the half-Singleton bound, proved in \cite{Cheng_2021} is true also when considering a code over $\Fq$ which is linear over a subfield of $\Fq$. 

\begin{thm}[Half–Singleton bound over a base field]
\label{thm:half_singleton_base}
Let \(\mathbb{E}\subset\mathbb{F}\) be finite fields and let \(\mathcal{C}\subseteq\mathbb{F}^{n}\) be an \(\mathbb{E}\)-linear code that can correct up to \(\delta n\) indels for some fixed \(\delta\in[0,1)\).  
Then,
\[
R(\mathcal{C}) \le \frac{1-\delta}{2}+\frac{1}{2n} \;.
\]
\end{thm}

%To present our results, we first define the following function, \(R_{\nicefrac{1}{2}}^{\mathrm{dn}}(\delta, q)\), which plays a key role in establishing the connection between \(R_{\nicefrac{1}{2}}(\delta,q)\) and \(R(\delta,q)\).
%\begin{defn}
%Let \({C} \subseteq \mathbb{F}_{q^2}^n\) be a code. 
%We say that \({C}\) is \emph{double-nonzero} if there exists a basis 
% \(\{b_1, b_2\}\) of \(\mathbb{F}_{q^2}\) over \(\mathbb{F}_q\) such that 
% every symbol \(\alpha = xb_1 + yb_2\) where \(x, y \in \mathbb{F}_q\) appearing in any codeword of \({C}\) 
% satisfies \(x = 0 \;\;\textit{if and only if}\;\; y = 0\). 
% Equivalently, every nonzero symbol in \({C}\) 
% has nonzero coordinates \(x, y \in \mathbb{F}_q\) in this basis.
% \end{defn}
% We define \(R_{\nicefrac{1}{2}}^{\mathrm{dn}}(\delta,q)\) in exactly the same way as \(R_{\nicefrac{1}{2}}(\delta,q)\), with two additional constraints. First, every code in the infinite family is double nonzero (as defined above). And second, the decoder can handle up to \(\delta  n\) indels and an unlimited number of deletions of zero symbols.
% All other aspects of the definition remain the same. From the definitions, it follows immediately that,
% \(
% R_{\nicefrac{1}{2}}^{\mathrm{dn}}(\delta,q) \;\leq\; R_{\nicefrac{1}{2}}(\delta,q).
% \)

%We prove the following result, which provides a lower and upper bound for \(R_{\nicefrac{1}{2}}^{\mathrm{dn}}(\delta, q)\) and \(R_{\nicefrac{1}{2}}(\delta,q)\).

We now present our constructions. The code constructions in this paper are explicit, have efficient, polynomial-time encoding and decoding algorithms. 
Our first construction presents half-linear codes with improved rate-error-correction tradeoff in the 
regime $\delta \leq 1/3$ (see \Cref{fig:combined-rate}). Specifically,

\begin{thm}[Informal, see \Cref{thm:half-lin-code-formal}]\label{thm:half-linear-code}
Let $\delta\geq 0$ and let $\varepsilon > 0$ be small enough. There exists an explicit and efficient code $\mathcal{C}$ defined over $\mathbb{F}_{q^2}$ that is linear over $\mathbb{F}_q$ for $q = \Theta(\varepsilon^{-4})$ such that $\mathcal{C}$ can correct efficiently $\delta$-fraction of indels and has rate $1/2 - \delta - \varepsilon$.

% Let $p$ be a prime and let $m$ be an positive integer. Set \(q = p^{2m}\) and let \(\delta \in (0,1)\). Then, 
% \[R_{\nicefrac{1}{2}}(\delta, q) \;\geq\; \frac{1}{2} \;-\; \delta \;-\; \Theta\bigl(q^{-\frac{1}{4}}\bigr).
% \]
\end{thm}

% Combining the previous two theorems with the results of \cite{con2022explicit}, we obtain that for $0\leq \delta \leq 1/3$
% \[
%     \frac12 - \delta - \Theta \left(q^{-1/4}\right)
%   \le R_{1/2}(\delta,q) \le \frac{1-\delta}{2}.
% \]
% and for $1/3 \leq \delta \leq 1/2$
% \[
%     \frac{1}{4} (1- \delta) - \Theta \left(q^{-1/4}\right)
%   \le R_{1/2} (\delta,q) \le \frac{1-\delta}{2}.
% \]

%We also prove the following result, which relates \(R_{\nicefrac{1}{2}}^{\mathrm{dn}}(\delta, q)\) and \(R(\delta, q)\).

% \begin{thm}\label{thm:half linear to full linear}
% For every \(\ell \in \mathbb{N}\), the following inequality holds, \(
% \frac{\ell}{\ell+1} \cdot R_{\nicefrac{1}{2}}^{\mathrm{dn}}(\delta, q) \;\leq\; R\Bigl(\frac{\delta}{2\ell+2}, q\Bigr).
% \)
% \end{thm}

%From the two previous theorems, we can derive the following lower bound for \(R(\delta, q)\).

Our next result provides an efficient construction of fully linear codes that can correct from indel errors. 
\begin{thm}[Informal, see \Cref{thm:full-linear-code-formal}]\label{thm:full linear results}
Let $\delta < 1/16$ and let $\varepsilon > 0$ be small enough. There exists an explicit and efficient linear code over $\Fq$ for $q = \Theta(\varepsilon^{-4})$ that can correct $\delta$-fraction of indels and has rate $1/2 - 2\sqrt{\delta} - \varepsilon$.

% Let \(q\) be a sufficiently large prime power, and let \(\delta \in (0,\frac{1}{16})\). Then, 
% \[
% R(\delta, q) \geq \frac{1}{2} - 2\sqrt{\delta} - \Theta\bigl(q^{-1/4}\bigr)\;.
% \]
\end{thm}
A graphical comparison of our work with~\cite{CST22} can be seen in \Cref{fig:combined-rate}. A key advantage of this paper over \cite{CST22} is that the rate of our linear codes can be arbitrarily close to $1/2$ while still correcting a constant fraction of indels. More formally, for any $\varepsilon > 0$ we construct efficient codes over alphabet of size $\Theta(\varepsilon^{-4})$ that have rate $1/2 - \varepsilon$ and can correct efficiently from $\Theta(\varepsilon^2)$ indels. 
We also note that in \cite{cheng2023linear} the authors constructed \emph{binary} linear codes
that have rate $1/2 - \varepsilon$ and can correct $\Theta(\frac{\varepsilon^{3}}{\log (1/\varepsilon)})$ indels. We leave it as an open question to construct, from our codes, a binary linear code with greater correction capability.

\subsection{Structure of the Paper}
The remainder of this paper is organized as follows.
In Section~\ref{sec:Preliminaries}, we provide preliminaries and background material necessary for our results.
In Section~\ref{sec:Construction of Half-Linear Codes}, we construct our half-linear code and prove Theorem~\ref{thm:half-linear-code} and in Section~\ref{sec:From Half-Linear Codes to Linear Codes}, we construct linear code that prove Theorem~\ref{thm:full linear results}. 
In Section~\ref{sec:abdel_extension} we prove Theorem~\ref{thm:half_singleton_base}, which establishes that the half-Singleton bound holds also for codes that are linear over a subfield.
Lastly, in Section~\ref{sec:Open_Problems_and_Future_Directions}, we discuss open problems and future research directions.

\section{Preliminaries}\label{sec:Preliminaries}
Throughout this work, \( q \) will denote a power of a prime, and \( \mathbb{F}_q \) will represent the finite field with \( q \) elements. All the codes constructed in this work will be defined over alphabets that are finite fields. Specifically, the codes will be constructed over the fields \( \mathbb{F}_{q^2} \) and \( \mathbb{F}_q \). 
Also, a vector will be denoted in bold, and its components will be indexed using subscripts. The position of a component in the sequence is referred to as its index. For example, the vector \( \bfc \) is represented as \( \bfc\) \(= (c_1, \dots, c_n) \), where \( c_i \) denotes the component of \( \bfc \) at index \( i \). We denote by \( \bfc [i, j] \) the subvector of \( \bfc \) consisting of the elements between indices \( i \) and \( j \), inclusive. Brackets \([ \, ]\) indicate inclusion of the endpoint, while parentheses \(( \, )\) indicate exclusion.
Throughout this paper, we shall move freely between representations of vectors as strings and vice versa. Namely, we view each vector $\bfv=(v_1, \ldots, v_n)\in \Fq^n$ also as a string by concatenating all the symbols of the vector into one string, i.e., $(v_1, \ldots, v_n) \leftrightarrow v_1 \circ v_2 \circ \cdots \circ v_n$. 

\subsection{Linear codes}
For a code \( C \subseteq \mathbb{F}_q^n \), \( d_\textup{H}(C) \) denotes the Hamming distance of the code, and \( R(C) \) denotes the rate of the code, defined as
\(
R(C) = \frac{\log_q(|C|)}{n}.
\)
As in \cite{CST22}, we will use AG-codes to construct our linear codes capable of correcting indels. These codes are defined over large (but constant) fields and are capable of correcting erasure and substitution errors efficiently.
\begin{thm}[AG codes \cite{tsfasman1982modular,skorobogatov1990decoding,kotter1996fast,stichtenoth2009algebraic}] \label{thm:AG-code}
Let $q = p^{2m}$ where $p$ is a prime and $m$ is a positive integer, and let $\delta \in (0, 1 - \frac{1}{\sqrt{q} - 1})$. There exists an explicit family of linear codes $\{C_i\}_{i=1}^{\infty}$ over $\mathbb{F}_q$ of lengths $\{n_i\}_{i=1}^{\infty}$ where $n_i\to\infty$ as $i\to\infty$, with minimal normalized Hamming distance $\delta$ and rate 
\(R(C_i) \geq 1 - \delta -\frac{1}{\sqrt{q} - 1}\).
Moreover, there is an efficient decoding algorithm that corrects, in time $O(n_i^3)$, $s$ substitutions and $e$ erasures, where $2s+e < \left(\delta - \frac{1}{\sqrt{q}-1}\right)\cdot n_i$.

% in polynomial time with respect to \(n\) up to the following combination of substitutions and erasures \(
% 2s + e <  d_\textup{H}(C) - \frac{n}{\sqrt{q} - 1} 
% ,\)
% where \( s \) represents the number of substitutions, and \( e \) represents the number of erasures.
\end{thm}

The authors of \cite{CST22} introduced the concept of half-linear codes, which we will also utilize in this work.
\begin{defn}[\cite{CST22}]
Let \( \mathbb{F}_q \) be a finite field. A code \( C \) over \( \mathbb{F}_{q^2} \) is called \emph{half linear} if it is closed under addition and scalar multiplication by elements of \( \mathbb{F}_q \).

\subsection{Levenshtein distance and self-matching sequences}
\end{defn}
Levenshtein introduced a metric between sequences, which we will utilize in the constructions presented in this work.
\begin{defn}
The \emph{Levenshtein distance} between two sequences \( \bfx \) and \( \bfy \), denoted by \( D_L(\bfx, \bfy) \), is the minimum number of insertions and deletions (indels), required to transform one sequence into the other.
\end{defn}

Our next definition, strongly related to the Levenshtein distance, is that of a longest common subsequence of two sequences. 
\begin{defn}\label{def:lcs}
Let \(\bfa = (a_1, \dots, a_m) \) and 
\( \bfb = (b_1, \dots, b_n) \) be two sequences. A \emph{Longest Common Subsequence (LCS)} of $\bfa$ and $\bfb$, denoted as $\textup{LCS}(\bfa, \bfb)$, is a sequence \( \bfc = (c_1, \dots, c_k) \) of maximal length such that \( \bfc \) is a subsequence of both \( \bfa \) and \( \bfb \). In other words, there exist indices \( 1 \leq i_1 < i_2 < \dots < i_k \leq m \) and \( 1 \leq j_1 < j_2 < \dots < j_k \leq n \) such that \( c_{\ell} = a_{i_{\ell}} = b_{j_{\ell}} \) for all \( 1 \leq \ell \leq k \).
\end{defn}

It is well known that $D_L(\bfa, \bfb) = |\bfa| + |\bfb| - 2|\textup{LCS}(\bfa, \bfb)|$ where $|\cdot|$ refers to the length of a sequence.
Our next definition is that of synchronization sequences. It was introduced in \cite{haeupler-survey2021synchronization} and served as a key ingredient in their construction of almost optimal nonlinear codes over large (but constant) alphabets correcting indels. 
\begin{defn}
A sequence \( \bfs \) of length \( n \) is called a \emph{\( \tau \)-self-matching sequence} if for every triple of indices \( 1 \leq i < j < k \leq n+1 \), it holds that \(
D_L(\bfs[i,j), \bfs[j,k)) > (1 - \tau)(k - i).
\)
\end{defn}
Our codes in the paper will also use synchronization sequences. The following theorem states that one can construct synchronization sequences in polynomial time over small (depending on $\tau$) alphabets.
\begin{thm}[Theorem 1.2, \cite{DBLP:conf/soda/ChengHLSW19}]
For every natural number \( n \) and every \( \tau \in (0,1) \), there exists a polynomial-time algorithm that constructs a \( \tau \)-self-matching sequence over an alphabet of size \( O(\tau^{-2}) \).
\end{thm}
The following corollary is just a translation of the previous theorem to the terminology of finite fields.  
\begin{cor}\label{cor:sync-exist}
Let \( \mathbb{F}_q \) be a finite field, and let \( n \) be a natural number.  
Then there exists a polynomial-time algorithm (in \( n \)) that constructs a \( \Theta\left(\frac{1}{\sqrt{q}}\right) \)-self-matching sequence of length \( n \), where all elements of the sequence are nonzero, i.e., they belong to \( \mathbb{F}_q^* \).
\end{cor}

\subsection{The Matching Algorithm of \cite{haeupler-survey2021synchronization}}\label{sec:hs-match}
For completeness, we give a brief overview of the algorithm \cite[Algorithm 1]{haeupler-survey2021synchronization} and its analysis. Throughout this paper, we call this algorithm \texttt{Match} which is given in \Cref{alg:match}.

\begin{algorithm}
\caption{\texttt{Match} \cite[Algorithm 1]{haeupler-survey2021synchronization}}
\begin{algorithmic}[1]
\Require $s$, $(c'_1, s'_1), \cdots, (c'_{m}, s'_{m})$
\Ensure $\bfy \in (\Fq \cup \{?\})^n$

\State $\bfs' \gets (s'_1,\ldots,s'_m)$.
\State $\textup{Pos} = (\perp, \ldots,\perp)$
\State $\bfy \gets (?,\ldots,?)$
\For {$i = 1$ to  $\floor{\frac{1}{\sqrt{\tau}}}$}
    \State Compute $\LCS(\bfs, \bfs')$ \label{lin:match-lcs-comp}
    \ForAll{Corresponding $s[i]$ and $\bfs'[j]$ in $\LCS(s, \bfs')$}
        \State $\textup{Pos}_j \gets i$ \label{lin:match-guess-pos}
    \EndFor
    \State Remove all elements of $\LCS(\bfs, \bfs')$ from $\bfs'$ \label{lin:match-remove-lcs}
\EndFor

\For{$i=1$ to $n$}\label{lin:match-actual-match}
    \If{ $|\{j \mid \textup{Pos}_j = i\}| = 1$}
        \State $y_i \gets c'_j$ for that unique $j$ where $\textup{Pos}_j = i$.
    \EndIf
\EndFor
\end{algorithmic}
\label{alg:match}
\end{algorithm}

Let $\bfs = s_1s_2\ldots s_n$ be a $\tau$-self-matching sequence. \Cref{alg:match} receives as input the self-matching string $\bfs$ and a vector of the form \(
((c'_1, s'_1), \dots, (c'_m, s'_m)) 
\)
which is obtained from $((c_1, s_1),\ldots, (c_n,s_n))$ after performing $\delta n$ indels.
The algorithm attempts to determine for each \( c'_i \) where it is located in the original word.
In line~\ref{lin:match-lcs-comp}, it computes the LCS of \( (s'_1, \dots, s'_m) \) with the self-matching string \(\bfs \). The LCS defines a correspondence between elements of the sequence \( (s'_1, \dots, s'_m) \) and the sequence \(\bfs\). If we match \( s'_i \) with \( s_j \), we are essentially guessing that \( c_i \) should be at position \( j \) (this happens in line~\ref{lin:match-guess-pos}). Then, on line~\ref{lin:match-remove-lcs}, the algorithm removes all the symbols from $s_1'\cdots s_m'$ that were matched. This process, computing the LCS, matching positions, and then removing the matched symbols, is repeated for $\floor{1/\sqrt{\tau}}$ times.

After the first loop, the vector $\textup{Pos}$ contains all the matches that were performed. In \cite[Lemma 2.2.]{haeupler-survey2021synchronization}, the authors proved that out of all the symbols that were \emph{not} deleted and sent properly, only at most $O(\sqrt{\tau} n)$ are matched incorrectly and furthermore, only at most $O(\sqrt{\tau} n)$ are not matched after the loop. However, note that there is a possibility that two symbols, say $s'_i$ and $s'_j$, are matched to the same $s_k$. 
More specifically, for every $j\in [n]$ there are two possibilities:
\begin{enumerate}
    \item Exactly one element is mapped to it,
    \item Zero or two or more elements are mapped to it.
\end{enumerate}
In the second loop in the algorithm, the actual matching is performed according to the following simple rule. For each $j\in [n]$, if exactly one element is mapped to it, then set $y_j$ to the respective matched code symbol and otherwise set $y_j$ to `$?$'.

The analysis in \cite[Theorem 2.3]{haeupler-survey2021synchronization} also describes how indels that are performed to $((c_1, s_1),\ldots, (c_n,s_n))$ affect the Hamming distance between $\bfc = (c_1, \ldots, c_n)\in \Fq^n$ and the output of the algorithm, $\widetilde{\bfy} = (y_1, \ldots, y_n)\in (\Fq\cup \{?\})^n$. 
Clearly, if no error is made and all the symbols are matched correctly, $\bfc = \bfy$. 
It was shown in~\cite{haeupler-survey2021synchronization} that
\begin{enumerate}
    \item Every deletion of a symbol, at the worst case, transforms a correctly matched symbol in $\bfy$ into a `$?$'.
    \item An insertion of a symbol, at the worst case, can turn a a correctly matched symbol into `$?$' or `$?$' into a substitution.
    \item The number of substitutions caused by wrong match plus the number of unmatched sent symbols is bounded by $O(\sqrt{\tau}n)$.
\end{enumerate}

Taking into consideration all of the above imperfections, \cite{haeupler-survey2021synchronization} proved the following statement.
\begin{lem}\cite[Lemma 2.2 and Theorem 2.3]{haeupler-survey2021synchronization} \label{lem:HS-lemma}
    Let $s_1\ldots s_n$ be a $\tau$-self-matching string. Let $\bfc = (c_1,\ldots,c_n)\in \Fq^n$ be a vector and assume that $((c'_1,s'_1), \ldots, (c'_m,s'_m))$ was obtained from $((c_1,s_1),\ldots, (c_n, s_n))$ after performing $\delta n$ indel errors.

    Then, applying the algorithm \texttt{Match} on $((c'_1,s'_1), \ldots, (c'_m,s'_m))$ outputs $\bfy = (y_1,\ldots,y_n)\in (\Fq \cup\{?\})^n$ such that $\bfy$ can be obtained from $\bfc$ by performing at most $e$ erasures and $t$ substitutions where $e + 2t \leq (\delta + 12\sqrt{\tau})n$.
\end{lem}

\section{Half-Linear Codes}\label{sec:Construction of Half-Linear Codes}

Since our construction of half-linear codes is highly inspired by~\cite[Construction 2.2]{CST22}, we begin with a brief description of that construction.
\subsection{The Linearization of \cite{CST22}}
\label{sec:half-linear-con-tamo-shpilka}
We first recall the construction of \cite{haeupler2017synchronization}.
Let \(C^{\textup{H}} \subseteq \mathbb{F}_q^n\) be a code capable of correcting substitutions and erasures. Then, define the code
\[
    C^{\text{ID}} = \{ (c_1, s_1), \dots, (c_n, s_n) \mid \bfc \in C^{\textup{H}}\}\;,
\]
where $\bfs = (s_1, \ldots, s_n)$ is a $\tau$-self-matching sequence over the alphabet $\Sigma_S$. This code, with a careful choice of $C^{\textup{H}}$, achieves a rate of $1-\delta - \tau$. Clearly, this is not a linear code.

Then, \cite{CST22} \emph{linearized} this code by turning each pair $(c_i,s_i) \in \Fq \times \Sigma_S$ into $(c_i, s_i\cdot c_i)\in \Fq \times \Fq$. Specifically, their code is defined as follows. 
  \begin{const}[Construction 2.2, \cite{CST22}] \label{const:cst22-const}
      Let $\delta_H>0$ and let $\tau$ be a small enough constant. Let $p$ be a prime such that $p=\Theta(\tau^{-2})$ and set $q = p^2 = \Theta(\tau^{-4})$. Define $C^{\textup{H}}$ to be a code of length $n$ that belongs to the family from \Cref{thm:AG-code} with normalized Hamming distance $\delta_H$ which has rate $\mathcal{R}(C^{\textup{H}}) = 1 - \delta_H - \tau$.
      Let $\bfs=s_1\ldots s_n$ be an $(\tau/24)^2$-self-matching string over $\Fq^{*}$. Define the code $C^{\textup{HL}}$ as
        \begin{equation}    \label{def:half-linear-code}
            C^{\textup{HL}} = \left\{ (c_1, s_1 \cdot c_1),  \dots, (c_n, s_n\cdot c_n) \mid \bfc \in C^{\textup{H}} \right\}\;.
        \end{equation}
  \end{const}

Note that this code is linear over $\Fq$ but not over $\Fq^2$, the finite field over which it is defined. 

\begin{rem}
    Remember that the self-matching string is the same in every codeword, and thus these symbols do not carry information. Now, note that in $C^{\textup{ID}}$, the alphabet size of the self-matching string can be made as small as we want compared to that of the code $C^{H}$. This allows one to achieve up to any small constant, the Singleton bound \cite{haeupler2017synchronization}.
    However, in $C^{\textup{HL}}$, since we are multiplying the self-matching symbol by the code symbol, and we want the code to preserve linearity (over $\Fq$), both of the symbols in the pair $(c_i, s_i \cdot c_i)$ are represented with $\log_2 q$ bits which implies a rate of at most $1/2$.
\end{rem}
  
The decoding algorithm of the code $C^{\textup{HL}}$ is given in \cite[Algorithm 1]{CST22} and is also provided here for completeness (see \Cref{alg:cst-decode-half-linear}). 
%Let $\bfc^{\textup{H}}\in C^{\textup{H}}$ and denote by $Z$ the number of zeros in $\bfc^{\textup{H}}$. 
% Denote by $\bfc^{\textup{ID}}$ and by $\bfc^{\textup{HL}}$ the respective codewords in $C^{\textup{ID}}$ and $C^{\textup{HL}}$. Let $\bfy$ be a word obtained from $\bfc^{\textup{HL}}$ by performing $\delta n$ indels. 

% \begin{enumerate}
%     \item Delete all zeros from $\bfy$ to obtain $\bfy'$.
%     \item Transform every symbol of $\bfy'$ into a pair $(c'_i,s'_i)$ (recall that the $i$th symbol in a codeword of $C^{\textup{HL}}$ is of the form $c_i + s_ic_i\alpha \in \mathbb{F}_{q^2}$ where $\{1, \alpha\}$ is a basis of $\mathbb{F}_{q^2}$ over $\Fq$).
%     Denote by $((c'_1,s'_1), \ldots, (c'_m, s'_m))$ the output of this step.
%     \item Run the \texttt{Match} algorithm to obtain $\widetilde{y}\in \Fq^n$
%     \item Decode $\widetilde{y}$ using the decoding algorithm of $C^{\textup{H}}$.
% \end{enumerate}

\begin{algorithm}
\caption{The decoding of $C^{\textup{HL}}$ according to \cite{CST22}}
\begin{algorithmic}[1]
\Require A word $\bfy \in (\mathbb{F}_{q^2})^{*}$.
\Ensure A message $\hat{m}\in \Fq^k$.
\vspace{0.5em}
\For{each nonzero symbol $y_i$ in $\bfy$} \label{lin:CST-zero-del}
    \State Extract $(c'_i,s'_i)$
\EndFor
\State Apply the \texttt{Match} algorithm to $((c'_1, s'_1), \ldots, (c'_m, s'_m))$ to obtain $\widetilde{\bfy} \in (\mathbb{F}_q \cup \{?\})^n$ \label{lin:CST-match-input}
\State Decode $\widetilde{\bfy}$ using the decoding algorithm for $C^{\textup{H}}$ to obtain $\hat{m}\in \Fq^k$
\end{algorithmic}
\label{alg:cst-decode-half-linear}
\end{algorithm}

We observe that \Cref{alg:cst-decode-half-linear} introduces ``new'' deletions, in line~\ref{lin:CST-zero-del} -- the zero symbols in $\bfy$ are deleted and not transformed to the \texttt{Match} algorithm nor to the decoder of $C^{\textup{H}}$. 
The reason for doing that in \cite{CST22} is that if a codeword of $C^{\textup{H}}$ contains a zero symbol, then the corresponding symbol in $C^{\textup{HL}}$ 
is $(0,0)$, and thus, we cannot extract the synchronization symbol.

We now argue about the rate-error-correction trade-off implied by this algorithm. Let $\bfc^{\textup{H}}\in C^\textup{H}$ and assume that it has $\zeta n$ zeros, where $\zeta\in [0,1]$.
Let $\bfc^{\textup{HL}} \in C^{\textup{HL}}$ be the corresponding transmitted codeword and let $\bfc^{\textup{ID}}\in C^{\textup{ID}}$ be the corresponding codeword in $C^{\textup{ID}}$. 
The vector $((c'_1, s'_1), \ldots, (c'_m, s'_m))$ in line~\ref{lin:CST-match-input} is obtained from $\bfc^{\textup{ID}}$ by performing $\delta n$ indels \emph{and} $\zeta n$ deletions.
Now, by \Cref{lem:HS-lemma}, $\widetilde{\bfy}$, obtained after running \texttt{Match}, is such that it can be obtained from $\bfc^{\textup{H}}$ after performing $e$ erasures and $t$ substitutions where $e + 2t \leq \delta n + \zeta n + \tau n/2$ (recall that the self-matching sequence has parameter $(\tau/24)^2$). 
Thus, if 
\[
\delta + \zeta + \tau/2 \leq \delta_H - \frac{1}{\sqrt{q} - 1} \;,
\]
then, the decoder of $C^{\textup{H}}$ succeeds, according to \Cref{thm:AG-code}. 
In \cite{CST22}, $\zeta$ was bounded by $1 - \delta_H$ and then by taking into consideration the parameters of \Cref{thm:AG-code} and that $\mathcal{R}(C^{\textup{HL}}) = \frac{1}{2}\mathcal{R}(C^{\textup{H}})$, we have that  $\mathcal{R}(C^{\textup{HL}}) = \frac{1}{4}(1- \delta) - \tau$.

% Note that $\bfy'$ can be obtained from $\bfc^{\textup{ID}}$ by performing $\delta n$ indels and $Z$ deletions. Thus, when \texttt{Match} is run with $\bfy'$ as input, it results in $\widetilde{\bfy}$ such that $\widetilde{\bfy}$ can be obtained from $\bfc^{\textup{H}}$ by performing $e$ erasures and $s$ substitutions where $2s + e = \delta n + Z+\Theta(\sqrt{\tau}n)$. 
% \cite{CST22} bounded $Z$ by $n - d_{\textup{H}}(C^{\textup{H}})$, resulting in the restriction that $R(C^{\textup{HL}}) \geq \frac{1}{4}(1- \delta) - \Theta(q^{-1/4})$.

\subsection{Decoding Algorithm}

We present our decoding algorithm, which is almost identical to \Cref{alg:cst-decode-half-linear}, yet the minor change in it gives a significant improvement in the rate-error-correction tradeoff. The algorithm is given in \Cref{alg:our-decode-half-linear} and its correctness is proved in \Cref{prop:half-lin-correcness}.

\begin{algorithm}
\caption{Our decoding algorithm for $C^{\textup{HL}}$}
\begin{algorithmic}[1]
\Require A word $\bfy \in (\mathbb{F}_{q^2})^{*}$.
\Ensure A message $\hat{m}\in \Fq^k$.
\vspace{0.5em}
\For{each nonzero symbol $y_i$ in $\bfy$} \ref{lin:CST-zero-del}
    \State Extract $(c'_i,s'_i)$
\EndFor
\State Apply the \texttt{Match} algorithm to $((c'_1, s'_1), \ldots, (c'_m, s'_m))$ to obtain $\widetilde{\bfy} \in (\mathbb{F}_q\in \{?\})^n$ \label{lin:our-match-input}
\State Replace all $?$ in $\widetilde{\bfy}$ with $0$ \label{lin:our-replace-questions}
\State Decode $\widetilde{\bfy}$ using the decoding algorithm for $C^{\textup{H}}$ to obtain $\hat{m}\in \Fq^k$ \label{lin:our-final-decoding}
\end{algorithmic}
\label{alg:our-decode-half-linear}
\end{algorithm}

\begin{prop} \label{prop:half-lin-correcness}
    Let $\tau,\delta > 0$ and let $\delta_H > 2\cdot( \delta + 12\tau)$. Let $C^{\textup{HL}}$ be the code defined in \Cref{const:cst22-const}.
    Assume that the codeword $\bfc^{\textup{HL}}\in C^{\textup{HL}}$ suffered from $\delta n$ indels and let $\bfy$ be the corrupted codeword. Then, on input $\bfy$, \Cref{alg:our-decode-half-linear} returns $\bfc$. Furthermore, the algorithm runs in time $O(n^3)$. 
\end{prop}

\begin{proof}
    Let $\bfc^{\textup{H}}$ be the codeword in $C^{\textup{H}}$ that corresponds to $\bfc^{\textup{HL}}$ and assume that $\bfc^{\textup{H}}$ has $\zeta n$ zeros. Further, let $\bfc^{\textup{ID}}$ be the corresponding codeword in $C^{\textup{ID}}$.
    Let $\widetilde{\bfy}$ be the string obtained after running line~\ref{lin:our-match-input}. 
    Observe that up to this point, the algorithm is identical to \Cref{alg:cst-decode-half-linear} and therefore, the vector $((c'_1,s'_1), \ldots, (c'_m,s'_m))$ can be obtained from $\bfc^{\textup{ID}}$ by performing at most $\delta n$ indels and $\zeta n$ deletions. 
    Therefore, by \Cref{lem:HS-lemma}, $\widetilde{\bfy}$, obtained after line~\ref{lin:our-match-input}, can be obtained from $\bfc^{\textup{H}}$ by performing $t$ substitutions and $e$ erasures where $e + 2t \leq (\delta + \zeta + 12\sqrt{\tau})n$.

    Now, instead of immediately decoding $\widetilde{\bfy}$, as is done in \Cref{alg:cst-decode-half-linear}, we replace all question marks in $\widetilde{\bfy}$ with the value $0$ (line~\ref{lin:our-replace-questions}). 
    For each such replacement, there are two options. If the erased symbol corresponds to a zero symbol, then the replacement operation ``fixed'' the erasure. Otherwise, the erasure corresponds to a nonzero symbol and therefore, this operation turned an erasure into a substitution. Denote by $e'$ the number of erasures that correspond to nonzero values. 
    Therefore, after line~\ref{lin:our-replace-questions} is performed, $\widetilde{\bfy}$ can be obtained from $\bfc^{\textup{H}}$ by performing only $t' = t + e'$ substitutions. 

    It remain to upper bound $t'$. For this purpose we recall the analysis of the algorithm \texttt{Match} and focus on the value of $\widetilde{\bfy}$ after line~\ref{lin:our-match-input} but before line~\ref{lin:our-replace-questions}. Remember that there are at most $O(\sqrt{\tau}n)$ substitutions that are caused by mismatches. 
    All other substitutions are caused by two indel operations (a deletion followed by an insertion). 
    As for the erasures, again, there are at most $O(\sqrt{\tau}n)$ of them that are caused by the algorithm (unmatched symbols). Every other erasure is caused by either an insertion or a deletion. 
    
    Therefore, the maximal number of substitution between $\bfc^{\textup{H}}$ and $\widetilde{\bfy}$, after line~\ref{lin:our-replace-questions}, is $\delta n + O(\sqrt{\tau}n)$. Indeed, at the worst-case scenario, every deletion in $\bfc^{\textup{HL}}$ turns into a substitution between $\bfc^{\textup{H}}$ and $\widetilde{\bfy}$. 
    Therefore, since $\delta n + 12 \sqrt{\tau} n < \frac{1}{2}\cdot \delta_H n$, we get that the decoder of $C^{\textup{H}}$ decodes correctly $\widetilde{\bfy}$ and thus, the algorithm succeeds.

    Note that we did not analyze the case where the zero vector is transmitted. 
    In this case, the number of symbols that are going to be the input to \texttt{Match} is at most $\delta n$ (the number of insertions the adversary can perform) which imlies that at most $\delta n$ symbols can be matched. Thus, after performing line~\ref{lin:our-replace-questions}, $\widetilde{\bfy}$ has Hamming weight at most $\delta n$ and the decoding on line~\ref{lin:our-final-decoding} succeeds. 
    % To this vector, at most $\delta n$ nonzero symbols are added, and after deleting the zero symbols we obtain a vector with up to $\delta n$ symbols.  
    % The MATCH algorithm can match at most $\delta n$ symbols, and therefore, after replacing the question marks with zero, we obtain a vector whose Hamming weight is at most $\delta n$.  
    % Hence, the decoding algorithm that runs in line~6 will return the zero vector as required.

    We are left to analyze the time complexity of the algorithm. Clearly, the entire loop just scans the input and thus runs in linear time. The \texttt{Match} algorithm takes $O(n^2)$ time and decoding of $C^{\textup{H}}$ takes $O(n^3)$, according to \Cref{thm:AG-code}. Thus, the total running time is dominated by $O(n^3)$.  
\end{proof}
%The improvement in the rate-error-correction trade-off is due to the fact that $C^{\textup{H}}$ can decode from unlimited number of zero erasures.
%Consider $\widetilde{\bfy}$ that is obtained at the third step of the algorithm in \Cref{sec:half-linear-con-tamo-shpilka}.  
%The total number of nonzero erasures and substitutions is at most $\delta n +O(\sqrt{\tau} n)$. Indeed, by \Cref{lem:HS-lemma}, at most $O(\sqrt{\tau} n)$ of the symbols that are not deleted are incorrectly guessed. Furthermore, by the analysis done in \cite[Theorem 2.3]{haeupler-survey2021synchronization} (see \Cref{sec:hs-match}), a deletion of a zero symbol turns the symbol into an `$?$' which is a zero erasure. Thus, they do not contribute to the creation of nonzero erasures or substitutions.

\subsection{Proving Theorem~\ref{thm:half-linear-code}}
In this section, we combine the previous pieces together to get the main theorem of this section. Note that this theorem is a more formal reformulation of \Cref{thm:half-linear-code}.
\begin{thm}\label{thm:half-lin-code-formal}
Let $\delta \in (0, \frac{1}{2})$ and let $\varepsilon > 0$ be small enough. There exists $q_0(\varepsilon) = \Theta(\varepsilon^{-4})$ such that for every $q>q_0(\varepsilon)$ that is a square the following holds. 
There exists an explicit family of codes $\{C_i\}_{i=1}^{\infty}$ over $\mathbb{F}_q^2$ of lengths $\{n_i\}_{i=1}^{\infty}$ where $n_i\to\infty$ as $i\to\infty$ such that for all $i$
\begin{itemize}
    \item $C_i$ is linear over $\Fq$.
    \item $C_i$ has rate $\cR \geq \frac{1}{2} - \delta -\varepsilon$.
    \item $C_i$ can correct $\delta$-fraction of indel errors in $O(n^3)$ time.
\end{itemize} 
\end{thm}
\begin{proof}

Let $\varepsilon > 0$ be a small enough constant and set $\tau = (\varepsilon/48)^2$. Let $\delta>0$ and set $\delta_H > 2(\delta + 12\tau)$. Now, let $C^{\textup{HL}}$ be the code defined in \Cref{const:cst22-const}.  
Giving \Cref{prop:half-lin-correcness}, it remains to compute the rate-error-correction tradeoff. 
Observe that the rate of $C^{\textup{HL}}$ is half of the rate of $C^{\textup{H}}$ and thus,
\begin{align*}
    \mathcal{R} &= \frac{1}{2}\cdot \left( 1 - \delta_H -\frac{1}{\sqrt{q} - 1}\right) \\
    &\geq \frac{1}{2}\cdot \left( 1 - 2\delta -24 \sqrt{\tau} - \frac{1}{\sqrt{q} - 1}\right) \\
    &\geq\frac{1}{2} - \delta - \varepsilon \;,
\end{align*}
where the last inequality follows by our choice of $\tau$ and since $1/(\sqrt{q} - 1) = \Theta(\varepsilon^{2})\leq \varepsilon/2$ for small enough $\varepsilon$.

\end{proof}

% Then, in order for $C^{\textup{H}}$ to correct from $\delta n +O(\sqrt{\tau} n)$ and unlimited number of zero erasures, we need that
% $\delta n +O(\sqrt{\tau} n) = \frac{d_{\textup{H}}(C^{\textup{H}})}{2} - \frac{n}{2(\sqrt{q}-1)} - 1$. Thus,
% \begin{align*}
% R(C^{\textup{HL}}) &\geq \frac{1}{2}\left(1 - \frac{d_\textup{H}(C^{\textup{H}})}{n} - \Theta(q^{-1/2})\right)\\
% &= \frac{1}{2}\left( 1 - 2\delta - O(\sqrt{\tau}) - O(q^{-1/2}) \right)
% \end{align*}
% where the equality follows for large enough $n$. The theorem follows by taking $\tau = \Theta(q^{-1/2})$ (see \Cref{cor:sync-exist}).

\section{From Half-Linear Codes to Linear Codes}\label{sec:From Half-Linear Codes to Linear Codes}
% In this section, we prove Theorem~\ref{thm:half linear to full linear}, restated below.
% \begin{reptheorem}{thm:half linear to full linear}
% For every \(\ell \in \mathbb{N}\), the following inequality holds, \(
% \frac{\ell}{\ell+1} \cdot R_{\nicefrac{1}{2}}^{\mathrm{dn}}(\delta, q) \;\leq\; R\Bigl(\frac{\delta}{2\ell+2}, q\Bigr).
% \)
% \end{reptheorem}

\subsection{Construction}
%We prove the theorem by a construction. Let \(q\) be a sufficiently large prime power, \( \delta \in (0,1)\) and \( \ell \in \mathbb{N}\).
%Let $C^{\textup{HL}}$ be a half-linear code of length $n$ over $\mathbb{F}_{q^2}$ satisfying the double non-zero property and can correct \( \delta  n \) indels and an unlimited number of zero deletions in polynomial time (in \( n \)).

Let $\ell > 0$ be an integer. We will present a construction that transforms $C^{\textup{HL}}$ into a linear code of length \( \roundDown{2n \cdot \frac{\ell+1}{\ell}} \).
We will denote this code by $C^{\ell}$ and show that it can correct from up to $\delta n/\ell$ indels.

We shall define two operations for this purpose. The first is $\textup{Flat}: (\Fq^2)^n \rightarrow \Fq^{2n}$ which, on a given input vector $((x_1,y_1), \ldots, (x_n, y_n))$, outputs $(x_1,y_1,x_2,y_2, \ldots, x_n,y_n)$.
The second operation is $\textup{Pad}_{\ell}: \Fq^{2n} \to \Fq^{\roundDown{\frac{2n}{\ell} \cdot (\ell+1)}}$ which, on a given input vector, after every $2\ell$ elements, adds two zeros.

% \begin{defn}\label{def:flat}
% Let \( \mathbb{F}_q \) be a finite field, and let \( \mathbb{U} = \{\bfu_1, \bfu_2\} \) be a basis for \( \mathbb{F}_{q^2} \) as a vector space over \( \mathbb{F}_q \).
% Define the function
% \(
% \emph{flat}_{\mathbb{U}}: \mathbb{F}_{q^2} \to \mathbb{F}_q^2
% \)
% as the function that takes an element of \( \mathbb{F}_{q^2} \) and returns its coordinate vector with respect to the basis \( \mathbb{U} \).
% Based on \( \text{flat}_\mathbb{{U}} \), define the following function,
% \(
% \emph{Flat}_{\mathbb{U}}: \mathbb{F}_{q^2}^n \to \mathbb{F}_q^{2n}
% \)
% such that for each vector, we apply \( \text{flat}_{\mathbb{U}} \) component-wise and then flatten the result into a vector of length \( 2n \).
% \end{defn}
\begin{exmp}   
For the vector,
\(
\bfc = ((0,2), (1,2), (0,1), (2,2)) \in (\mathbb{F}_3^2)^4\), we get that
\[
\textup{Pad}_3(\textup{Flat}(\bfc)) = (0,2,1,{\color{blue}0,0,}2,0,1,{\color{blue}0,0},2,2)
\]
\end{exmp}
% Let \( \mathbb{U} \) be a basis for which \( C^{\textup{HL}} \) has the double non-zero property. 
\begin{const} \label{const:our-lin-const}
    Let $\delta_H>0$ and let $\tau$ be a small enough constant. Let $C^{\textup{HL}}$ be the code defined in \Cref{const:cst22-const} with $\tau, \delta$.
    We define,
    \begin{equation} \label{eq:lin-code}
        C^{\ell} = \left\{\, \textup{Pad}_{\ell}\bigl(\textup{Flat}(\bfc)\bigr) \;\big|\; \bfc \in C^{\textup{HL}} \,\right\}.
    \end{equation}
 \end{const}
 It holds that \( C^{\ell} \) is a linear code and has length
\(
 \lfloor 2n \cdot \frac{\ell+1}{\ell} \rfloor\;.
 \)
Also, observe that the function \( \text{Pad}_{\ell} \circ \text{Flat}_{U} \) is injective; therefore  
\(
|C^{\ell}| = |C^{\textup{HL}}|.
\)  
In the next section, we show that \( C^{\ell} \) can correct $\delta n/\ell$ indels, thereby completing the proof of the theorem. 

We shall need the following trivial claim which establishes that all the runs of the symbol $0$ in codewords of $C^{\ell}$ are of even length.
\begin{claim} \label{clm:double-zero-prop}
    For every codeword $\bfc\in C^{\ell}$, every run of the symbol zero is of even length. 
\end{claim}
\begin{proof}
    Recall that every symbol of every codeword $\bfc^{\textup{HL}}\in C^{\textup{HL}}$ is of the form $(c_i, s_i \cdot c_i)$ where $s_i$ is nonzero. Thus, it must be that either both are zero or both are nonzero. 
    Therefore, in $\textup{Flat}(\bfc^{\textup{HL}})$ there cannot be a run of zeros of odd length. 
    Clearly, applying $\textup{Pad}_{\ell}$ adds only runs of zeros of even length.
\end{proof}
\subsection{Decoding Algorithm}
To describe our decoding algorithm, we need to define a \emph{segmentation} of a sequence.
\begin{defn}\label{def:segmentation}
    Let \( \bfy \in \mathbb{F}_q^m\). write $\bfs$ as \[
    \bfs = 0^{d_0} \circ \bfy_1 \circ 0^{d_1} \circ \bfy_2 \circ \cdots \circ 0^{d_t} \circ \bfy_t \circ 0^{d_t},
    \]
    where each $\bfy_j$ is a maximal length contiguous subsequence that does not contain zeros. Then, \(\textup{seg}(\bfy) = (\bfy_1, \dots, \bfy_t)\). 
    Denote by $m_j$ the length of every $\bfy_j$. Each \( \bfy_i \) is called a \emph{window}, and the sequences of zeros between windows are called \emph{delimiters}.
\end{defn}

\begin{exmp} 
Let \( \bfy = (1, 1, 1, 0, 2, 1, 3, 0, 0, 0, 1) \) be a sequence over \( \mathbb{F}_5 \). Then,
\(
\text{seg}(\bfy) = \big((1, 1, 1), (2, 1, 3), (1)\big).
\)
This means that,
\(
\bfy_1 = (1, 1, 1), \bfy_2 = (2, 1, 3), \bfy_3 = (1).
\)
\end{exmp}

\begin{algorithm}
\caption{Decode $C^{\ell}$}
\begin{algorithmic}[1]
\Require A vector $\bfy$
\Ensure A codeword $\bfc\in C^{\textup{HL}}$
\State $L\gets$ empty list
\State Compute $\text{seg}(\bfy) = (\bfy_1, \ldots, \bfy_t)$ \label{lin:lin-dec-segmetation}
\For{$j = 1$ to $t$} \label{lin:lin-dec-main-loop}
    \If{$|\bfy_j| > 2\ell$ or $|\bfy_j|$ is odd} \label{lin:lin-dec-odd-or-big}
        \State Continue
    \EndIf
    \State Split $\bfy_j = (y_{j,1}, \ldots, y_{j,|\bfy_j|})$ into pairs:
    \[
        (y_{j,1},y_{j,2}),(y_{j,3},y_{j,4}), \ldots, (y_{j, |\bfy_j|-1}, y_{j, |\bfy_j|})\;.
    \] \label{lin:lin-dec-pair-gen}
    \State Append these pairs to the end of $L$.
\EndFor
\State Apply the decoding algorithm of $C^{\textup{HL}}$ on $L$.\label{lin:lin-dec-final-decode}
\end{algorithmic}
\label{alg:linear-decoder}
\end{algorithm}

% \begin{enumerate}
%     \item Compute
%     \(
%     \text{seg}(\bfy) = (\bfy_1,  \ldots, \bfy_t).
%     \)
%     \item filter out each \( \bfy_j \) that has an odd length or a length exceeding \( 2\ell \).
    
%     Let the remaining windows be denoted by \( (\bfh_1, \ldots, \bfh_r) \), and define an empty list \( L \).
%     \item For each \( j \) from \( 1 \) to \( r \), we split \( \bfh_j \) into pairs and convert each pair into an element of \( \mathbb{F}_{q^2} \) using the basis \( \mathbb{U} \), then append this element to the end of \( L \).
%     \item Apply the decoding algorithm of $C^{\textup{HL}}$ on $L$.
% \end{enumerate}

Our decoding algorithm is given in \Cref{alg:linear-decoder}. In the following proposition, we show its correctness.
\begin{prop} \label{prop:lin-code-alg-correctness}
    Let $\tau$ be a small enough constant. Let $\delta>0$ and define $\delta_H > 2\cdot (\delta\cdot \ell + 12\tau)$. Let $C^{\ell}$ be the code defined in \Cref{const:our-lin-const} with $\tau$ and $\delta_H$. 
    
    Let $\bfc\in C^{\ell}$ be a codeword and assume that $\bfy$ is obtained from $\bfc^{\ell}$ after performing $\delta n$ indels. 
    Then, on input $\bfy$, \Cref{alg:linear-decoder} returns $\bfc^{\textup{HL}}$ in time $O(n^3)$.

    %Then, the sequence \( L \) can be derived from \( \bfc^{\ell} \) by performing \( \delta  n \) indels and deleting all the zero symbols.

\end{prop}
\begin{proof} 

For the zero word, there are at most
$\delta n / \ell$ insertions
and at most
$\delta n / \ell$ deletions.
During the execution of the algorithm, all zeros are ignored,
and for each of the inserted symbols, some are discarded in line~4,
while the remaining ones are combined such that every two symbols become a single symbol added to $L$.
Therefore, $L$ contains at most
$\delta n / (2\ell)$
symbols.
In the execution of the decoding algorithm in line~10,
the \textsc{Match} algorithm will return
a word with Hamming weight at most
$\delta n / (2\ell)$,
and in this case the decoding algorithm will return the zero word as required.

For the non zero word. Let $\bfc^{\ell}$ be the transmitted codeword and let $\bfc^{\textup{HL}}$ be the corresponding codeword in $C^{\textup{HL}}$. Our goal is to show that just before performing line~\ref{lin:lin-dec-final-decode} it holds that $D_L(L, \bfc^{\textup{HL}}) \leq \delta \ell n$. This implies that the decoding algorithm of $C^{\textup{HL}}$ when given as input $L$, succeeds according to \Cref{prop:half-lin-correcness}.
    
    Let $\bfs = \textup{seg}(\bfc^{\ell})$. The adversary, who knows the decoding algorithm, will try to cause as much ``damage'' as he can with his budget of $\delta n$ indel errors.
    In the following, we list several operations the adversary can do and analyze their outcome.

    First, the adversary can cause an odd number of indels inside a window of $\bfs$. In this case, the length of the window becomes odd, and the algorithm (line~\ref{lin:lin-dec-odd-or-big}) will discard it. This causes our algorithm to effectively delete at most $\ell$ nonzero symbols from $\bfc^{\textup{HL}}$. Clearly, the most ``economic'' way for the adversary to achieve this effect is to perform a single indel.

    Second, the adversary can perform an even number of indels inside a window of $s$. More specifically, assume that there were $I$ insertions and $D$ deletions to this window. We consider two subcases
    \begin{itemize}
        \item First assume that $I\geq D$.
        In this case, in the worst-case scenario, all the symbols of the window in $\bfc^{\textup{HL}}$ suffer from substitutions. 
        Indeed, the adversary can easily damage the pair synchronization of the symbols (e.g., delete the first symbol in the window and insert a symbol at the end of the window) and then the algorithm (in line~\ref{lin:lin-dec-pair-gen}) will form pairs such that every pair is misaligned and thus wrong.
        Also, observe that if $I\geq D$, then
        $(I-D)/2$ symbols are inserted to this window. In total, in this case, the adversary, by performing $D+I$ indels, deleted at most $\ell$ pairs and inserted $\ell + (I-D)/2$ new pairs. 
   \item 
        Second, if $I < D$, then the number of pairs that are formed in line~\ref{lin:lin-dec-pair-gen} is $r - (D-I)/2$ where $r$ is the number of pairs in the original window of $\bfs$. 
        Thus, in this case, by performing $D+I$ indels, the adversary deletes at most $\ell$ pairs and inserts $\ell - (D-I)/2$. 
    \end{itemize}
    Concluding these two subcases, we see
    the ratio between the number of ``pair'' errors and the number of indels that the adversary performs is maximized when $D = I = 1$.

    Third, the adversary can delete delimiters of $\bfc^{\ell}$. This can cause a merge of two adjacent windows. According to our construction, the size of a window is at most $2\ell$, and thus, merging two windows creates a new window of length at most $4\ell$. 
    By our algorithm, any window of length greater than $2\ell$ is not considered for decoding (line~\ref{lin:lin-dec-odd-or-big}).
    Therefore, by deleting say $r$ consecutive delimiters, our algorithm deletes at most $(r+1) \cdot 2\ell$ pairs of symbols from $\bfc^{\textup{HL}}$. 
    Now, by \Cref{clm:double-zero-prop}, deleting a delimiter requires at least two deletions. 
    Thus, one can verify that the best case for the adversary is to delete a single delimiter of length two and then the algorithm deletes $2\ell$ pairs from $\bfc^{\textup{HL}}$. 
    % \begin{itemize}
    %     \item If an odd number of indels occurs inside a window of $\bfs$, then the length of the window becomes odd. 
    %     This causes our algorithm to effectively delete at most $\ell$ nonzero symbols from $\bfc^{\textup{HL}}$.
    %     \item If an even number of indels occurs within a window then our algorithm will perform step $3$ on this window which can result in
    %     at most \( \ell \) substitutions to $ \bfc^{\textup{HL}}$. The number of indels in this case is therefore $2\ell$ in $\bfc^{\textup{HL}}$.
    %     \item If a delimiter is deleted in $\bfc^{\ell}$, this causes a merge of two adjacent windows. Recall that we discard windows of length larger than $2\ell$. Since the size of a window is at most $2\ell$, merging two windows creates a new window of length at most $4\ell$. Therefore, removing such a window results in at most $2\ell$ deletions to $\bfc^{\textup{HL}}$.
    %     Now, by the double non-zero property, every delimiter must be of size at most $2$ and thus to perform $2\ell$ deletions to $\bfc^{\textup{HL}}$ the adversary needs to perform at least $2$ deletions to $\bfc^{\ell}$.
    % \end{itemize}
    % The proposition follows by easily verifying that any other type of indel to $\bfc^{\ell}$ causes less than $\ell$ indels to $\bfc^{\textup{HL}}$.

    Concluding all the above cases, we get every indel the adversary does to $\bfc^{\ell}$ can cause at most $\ell$ indels to $\bfc^{\textup{HL}}$. Thus, since the adversary can perform at most $\delta n$ indels, it holds that $D_L (\bfc^{\textup{HL}}, L)\leq \ell \cdot \delta n$ and line~\ref{lin:lin-dec-final-decode} succeeds and $\bfc^{\textup{HL}}$ is returned.

    We are left to show the time complexity of the algorithm. Clearly, in line~\ref{lin:lin-dec-segmetation}, computing the segmentation of $\bfy$ takes linear time. The loop on line~\ref{lin:lin-dec-main-loop} also takes linear time since we just split each window into pairs.
    Thus, the total time complexity is $O(n) + T_{\textup{dec}}(C^{\textup{HL}}) = O(n^3)$ where the equality is due to \Cref{prop:half-lin-correcness}.
\end{proof}

\begin{thm} \label{thm:C-l-rate-error-trd}
    Let $\ell > 0$ be an integer.
    Let $\delta \in (0, 1/16)$ and let $\varepsilon > 0$ be small enough. There exists $q_0(\varepsilon) = \Theta(\varepsilon^{-4})$ such that for every $q>q_0(\varepsilon)$ that is a square, the following holds. 
There exists an explicit family of linear codes $\{C^{\ell}_i\}_{i=1}^{\infty}$ over $\mathbb{F}_q^2$ of lengths $\{n_i\}_{i=1}^{\infty}$ where $n_i\to\infty$ as $i\to\infty$ such that for all $i$
\begin{itemize}
    \item $C^{\ell}_i$ has rate $\cR \geq \frac{\ell}{\ell + 1}\left( \frac{1}{2} -  2(\ell+1)\cdot \delta\right) -\varepsilon$.
    \item $C^{\ell}_i$ can correct $\delta$-fraction of indel errors in $O(n^3)$ time.
\end{itemize} 
\end{thm}
\begin{proof}
    Let $\varepsilon>0$ be a small enough constant, and set $\varepsilon' = \frac{\ell + 1}{\ell} \cdot \varepsilon$. Moreover, set $\tau = (\varepsilon'/48)^2$ and set $\delta_H > 2\cdot (\ell \cdot \delta' + 12\tau)$ for some $\delta'$. Define the code $C^{\ell}$ according to \Cref{const:our-lin-const} with $\tau$ and $\delta_H$. 
    According to \Cref{prop:lin-code-alg-correctness}, the code can correct any $\delta' n$ indel errors. However, since the length of $C^{\ell}$ is $2n \cdot \frac{\ell+1}{\ell}$, the actual fraction of deletions $C^{\ell}$ can correct is $\delta := \frac{\ell}{2(\ell+1)} \delta'$. Furthermore, note that the construction of $C^{\ell}$ in \Cref{const:our-lin-const} is defined using the code $C^{\textup{HL}}$ (\Cref{const:cst22-const}) with the same parameters $\delta_H$ and $\tau$. Thus, by \Cref{thm:half-linear-code}, the rate of $C^{\textup{HL}}$ is $\frac{1}{2} - \ell \cdot \delta' - \varepsilon'$. 

    Now, since $|C^{\textup{HL}}| = |C^{\ell}|$, we have
    \[
        \mathcal{R}(C^{\ell}) = \frac{\ell}{\ell+1} \cdot \mathcal{R}(C^{\textup{HL}}) \geq \frac{\ell}{\ell + 1}\left( \frac{1}{2} -  \ell\cdot \delta' -\varepsilon' \right) \;.
    \]
    The theorem follows by our choice of $\varepsilon$ and the relation between $\delta'$ and $\delta$.
\end{proof}
% We now compute the rate-error-correction tradeoff, as a function of $\ell$. 
% We take into consideration that  $|C^{\textup{HL}}| = |C^{\ell}|$ and that $C^{\ell}$ can correct from $\delta = \delta' /\ell$ fraction of indels. We have
% \[
% \mathcal{R}(C^{\ell}) = \frac{\ell}{\ell+1} \cdot \mathcal{R}(C^{\textup{HL}}) \geq \frac{\ell}{\ell + 1}\left( \frac{1}{2} -  \ell\cdot \delta -\varepsilon \right)\;.
% \]

% The length of the new code is given by \(\roundUp{\frac{n}{\ell} \cdot (2\ell + 2)},
% \)
% and, as we have shown, it can correct up to \(\frac{\delta}{\ell} n\)
% indels. Therefore, the correction capability divided by the code length is
% \(
% \frac{\delta}{\ell}  n\Big/\roundUp{{\frac{n}{\ell}  (2\ell + 2) }}.
% \)
% For code length tending to infinity, this expression simplifies to \(\frac{\delta}{2\ell+2}\).
% This completes the proof of Theorem~\ref{thm:half linear to full linear}

\subsection{Proving \Cref{thm:full linear results}}\label{sec:high rate regime}
In this section, we prove Theorem~\ref{thm:full linear results}. We restate it below a bit more formally as a corollary of \Cref{thm:C-l-rate-error-trd}.
\begin{cor} \label{thm:full-linear-code-formal}
Let $\delta \in (0, 1/16)$ and let $\varepsilon > 0$ be small enough. There exists $q_0(\varepsilon) = \Theta(\varepsilon^{-4})$ such that for every $q>q_0(\varepsilon)$ that is a square, the following holds. 
There exists an explicit family of linear codes $\{C_i\}_{i=1}^{\infty}$ over $\mathbb{F}_q^2$ of lengths $\{n_i\}_{i=1}^{\infty}$ where $n_i\to\infty$ as $i\to\infty$ such that for all $i$
\begin{itemize}
    \item $C_i$ has rate $\cR \geq \frac{1}{2} - 2\sqrt{\delta} -\varepsilon$.
    \item $C_i$ can correct $\delta$-fraction of indel errors in $O(n^3)$ time.
\end{itemize} 
\end{cor}

\begin{proof}
According to \Cref{thm:C-l-rate-error-trd}, we have
\begin{align*}
\mathcal{R}(C^{\ell}) &\geq \frac{\ell}{\ell + 1}\left( \frac{1}{2} -  2(\ell + 1)\cdot \delta \right)  -\varepsilon \\
&= \left(1 - \frac{1}{\ell + 1}\right) \cdot \left( \frac{1}{2} -  2(\ell + 1)\cdot \delta  \right) - \varepsilon \\
&= \frac{1}{2} + 2\delta - 2\delta(\ell+1)-\frac{1}{2(\ell+1)} - \varepsilon \;.
\end{align*}

Now, by setting $\ell$ to be the integer such that $\frac{1}{2\sqrt{\delta}} \leq \ell+1 \leq \frac{1}{2\sqrt{\delta}} + 1$, we get the desired rate.

\end{proof}

% XXXXX\\\\\\
% Let \(\ell\) be a natural number to be determined later. From Theorem~\ref{thm:half linear to full linear}, we have\(
% \frac{\ell}{\ell+1} \cdot R_{\nicefrac{1}{2}}^{\mathrm{dn}}(2\delta(\ell+1), q) \leq R(\delta, q).\)
% From Theorem~\ref{thm:half-linear-code}, it follows that
% \(
% R_{\nicefrac{1}{2}}^{\mathrm{dn}}(2\delta(\ell+1), q) \geq \frac{1}{2} - 2\delta(\ell+1)-\Theta\bigl(q^{-1/4}\bigr).
% \)
% Combining these results gives
% \(
% R(\delta, q) \geq \frac{\ell}{\ell+1} \cdot \left(\frac{1}{2} - 2\delta(\ell+1)\right) - \Theta\bigl(q^{-1/4}\bigr).
% \)
% Now, observe that
% \(
% \frac{\ell}{\ell+1} = 1 - \frac{1}{\ell+1}.
% \)
% Substituting this into the inequality yields
% \(
% R(\delta, q) \geq \left(1 - \frac{1}{\ell+1}\right) \cdot \left(\frac{1}{2} - 2\delta(\ell+1)\right) - \Theta\bigl(q^{-1/4}\bigr).
% \)
% Expanding and simplifying, we get
% \(
% R(\delta, q) \geq \frac{1}{2} + 2\delta-2\delta(\ell+1)-\frac{1}{2(\ell+1)}- \Theta\bigl(q^{-1/4}\bigr).
% \)
% Next, we choose \(\ell\) such that
% \(
% \frac{1}{2\sqrt{\delta}} \leq \ell+1 \leq \frac{1}{2\sqrt{\delta}} + 1.
% \)
% Under this choice, we have
% \(
% 2\delta(\ell+1) \leq \sqrt{\delta} + 2\delta, \quad \text{and} \quad \frac{1}{2(\ell+1)} \leq \sqrt{\delta}.
% \)
% Using these bounds, we find that
% \(
% 2\delta - 2\delta(\ell+1) - \frac{1}{2(\ell+1)} \geq -2\sqrt{\delta}.
% \)
% Substituting this back into the inequality for \(R(\delta, q)\), we conclude,
% \(
% R(\delta, q) \geq \frac{1}{2} - 2\sqrt{\delta} - \Theta\bigl(q^{-1/4}\bigr).
% \)
% \end{proof}

%====================================================
\section{Half–Singleton Bound for Base‑Field Codes}
\label{sec:abdel_extension}
In this section, we prove Theorem~\ref{thm:half_singleton_base}. For convenience, we restate it below.

\begin{reptheorem}{thm:half_singleton_base}
Let \(\mathbb{E}\subset\mathbb{F}\) be finite fields and let \(\mathcal{C}\subseteq\mathbb{F}^{n}\) be an \(\mathbb{E}\)-linear code that can correct up to \(\delta n\) indels for some fixed \(\delta\in[0,1)\).  
Then,
\[
R(\mathcal{C})\;\le\;\frac{1-\delta}{2}+ \frac{1}{2n},
\]
\end{reptheorem}

To prove this theorem, we will closely follow the proof of \cite{abdel2007linear}. We start by borrowing a lemma from \cite{abdel2007linear} which gives a sufficient condition for a linear code to fail to correct even one deletion. In \cite{abdel2007linear} this lemma was proved for linear codes; however, the proof used only closure under addition and thus the claim holds also for codes that are linear over a subfield. For completeness, we shall also include the proof.

\begin{lem} \cite[Lemma 2]{abdel2007linear} \label{lem:cond-for-not-del-correct}
    Let $C$ be an $[n,k]_q$ linear code over $\mathbb{E}$. Then, $C$ cannot correct a single deletion if it contains a nonzero codeword $\bfx=(x_1,\ldots,x_n)$ such that $\bfc = (c_1,\ldots,c_n)$ is also a codeword where $c_i = \sum_{j=1}^{i-1}x_j$ for all $i\in [n]$.
\end{lem}
\begin{proof}
The authors of the cited paper proved the following lemma for the case of a linear code. We adapt it here to our setting of a linear code over a base field.
Lemma, Let $\mathbb{E}, \mathbb{F}$ be fields such that $\mathbb{F}$ is an extension field of $\mathbb{E}$. Consider a code 
$C \subseteq \mathbb{F}$
that is closed under addition and under scalar multiplication by elements of $\mathbb{E}$. Then $C$ cannot correct a deletion if and only if $C$ contains a codeword 
$\bfc = (c_1,\ldots,c_n)$
such that for some $1 \leq u \leq u' \leq n$ and some field element $a \in F$, the vector 
$x = (x_1,\ldots,x_n)$
defined by
\[
x_i = \begin{cases}
0 & \text{for } i \in [1,u) \cup (u',n], \\
c_{i+1}-c_i & \text{for } i \in [u,u'), \\
a & \text{for } i=u',
\end{cases}
\]
is a nonzero codeword.
In one direction, assume that $C$ cannot correct a deletion. Then there exist distinct codewords $\bfc, \bfc' \in C$ such that deleting coordinate $u$ from $\bfc$ and deleting coordinate $u'$ from $\bfc'$ yield the same word, for some $1 \leq u \leq u' \leq n$.  
Let $a = c_u = c'_{u'}$ and define $\bfx = \bfc - \bfc'$. Since $-1$ belongs to every field, closure under scalar multiplication by $\mathbb{E}$ implies that $-\bfc' \in C$. By closure under addition, we conclude that $\bfx \in C$. Moreover, $\bfx \neq 0$ because $\bfc \neq \bfc'$, as required.
Conversely, suppose that $C$ contains a codeword $\bfc$ such that for some $1 \leq u \leq u' \leq n$ and some $a \in \mathbb{F}$, the vector $\bfx$ defined above is a nonzero codeword. Set $\bfc' = \bfc - \bfx$. By the same argument, $\bfc' \in C$ and $\bfc' \neq c$. We then obtain two distinct codewords $\bfc, \bfc'$ such that deleting coordinate $u$ from $\bfc$ and coordinate $u'$ from $\bfc'$ yields the same word. Hence $C$ cannot correct a deletion.
Finally, note that by applying the lemma with $u=1$, $u'=n$, $c_1=0$, and $a = x_n$, we obtain the desired conclusion.

\end{proof}

\begin{prop} \label{prop:sing-indel-ub}
    Let \(\mathbb{E}\subset\mathbb{F}\) be finite fields and set \(\ell\coloneqq[\mathbb{F}:\mathbb{E}]\).  
Let \(\mathcal{C}\subseteq\mathbb{F}^{n}\) be an $[n,k]_q$ code that is linear over $\mathbb{E}$. Then, if the rate of $C$ is strictly bigger than $1/2$, the code cannot correct a single indel.
\end{prop}
\begin{proof}
    Every vector in $\Fq^n$ can be represented as a vector of length $\ell \cdot n$ over $\mathbb{E}$. Since the rate of the code is larger than $1/2$, we have
    \begin{align*}
        \frac{1}{2} < \frac{\log_{|\Fq|}|C|}{n} = \frac{\log_{|\mathbb{E}|}|C|}{\ell \cdot n}\;.
    \end{align*}
    We will view $C$ as a linear code over $\mathbb{E}$ with length $\ell\cdot n$ and dimension $k:=\log_{\mathbb{E}}|C|$. As such, it has a parity-check matrix $H \in \mathbb{E}^{(\ell n-k)\times \ell n}$. From \Cref{lem:cond-for-not-del-correct}, we know that $C$ cannot correct a single indel if there is $\bfx \in C$ and $\bfc\in C$ such that $c_i = \sum_{j=1}^{i-1}x_j$ for $i\in [n]$. 

    We prove that there are two such codewords in $C$. Indeed, from $\bfx \in C$, we have that $\bfx \cdot H = \bfzero$ which gives us $\ell n - k$ linear equations in the variables $x_1,\ldots, x_n$. From $\bfc \in C$ we get $\bfc\cdot H = \bfzero$, which gives another $\ell n - k$ linear equations (the equations are $\langle \sum_{j=1}^{i-1}x_j, \bfh_i \rangle$ where $h_i$ is the $i$th row of $H$). In total, we have $2 (\ell n - k)$ homogeneous linear equations in $\ell \cdot n$ variables. 
    Since by assumption, $k >\ell n/2$, it must be that this system has a nontrivial solution $\bfx = (x_1, \ldots, x_n)$. This nontrivial solution gives rise to two codewords $\bfx$ and $\bfc$ of the form described in \Cref{lem:cond-for-not-del-correct} and therefore, the code cannot correct an indel error. 
\end{proof}

Our next goal is to extend this proposition and prove \Cref{thm:half_singleton_base}. We shall follow the exact steps of the proof \cite{Cheng_2021}.

\begin{proof}[Proof of \Cref{thm:half_singleton_base}]
    Let $C$ be a code defined over $\mathbb{F}_q$ but is linear over $\mathbb{E}$. Furthermore, assume that the length of $C$ is $n$ and that $C$ can correct from $\delta n$ indels. 
    Delete from all codewords of $C$, the first $\delta n - 1$ symbols. 
    The resulting code, $C'$, has length $(1- \delta)n + 1$. 

    Now, observe that $C'$ can still correct a single indel. Furthermore, it is linear over $\mathbb{E}$ and it holds that $|C| = |C'|$ since otherwise, it would contradict the assumption that $C$ can correct $\delta n$ indels. According to \Cref{prop:sing-indel-ub}, we have
    \[
        \mathcal{R}(C')  = \frac{\log_{|\mathbb{F}_q|}|C'|}{(1- \delta)n + 1} \leq \frac{1}{2} \;,
    \]
    which implies that
    \[
    \mathcal{R}(C) = \frac{\log_{|F_q|}|C|}{n} = \frac{\log_{|\mathbb{F}_q|}|C'|}{n} \leq \frac{1 - \delta}{2} + \frac{1}{2n}
    \]
    and the theorem follows.
\end{proof}

\section{Summary and Open Problems}\label{sec:Open_Problems_and_Future_Directions}
In this paper, we continued the study on the performance of linear codes in the presence of indel errors.
We first proved that the Half-Singleton bound holds for codes that are linear over a subfield, i.e.\ codes closed under both addition and scalar multiplication by subfield elements.
Building on these findings, we conclude by highlighting several open problems.

\begin{itemize}
\item \textbf{Additive codes and the Half-Singleton bound.}
A natural next question is whether the same bound remains valid for codes that are merely \emph{additive}--closed only under addition--or whether one can construct additive codes that surpass it. 
Formally, fix the residue ring $\mathbb{Z}_r$ and let $C \subseteq \mathbb{Z}_r^{n}$ be an additive code that corrects $\delta n$ indels.  What is the maximal rate attainable by such a code?  When $r$ is a prime power the Half-Singleton bound continues to apply, but for composite $r$ that are not prime powers the problem is still open.

\item \textbf{Closing the gap to the Half-Singleton bound.}  
Our explicit half-linear and fully linear constructions already come closer to
the Half-Singleton bound, yet a non-negligible gap remains. 
A straightforward task is to further push the constructions and achieve an efficient construction that achieves the half-Singleton bound.

\item \textbf{Binary codes with high rate.}  
Another interesting question is to get binary codes from our codes in such a way that preserves the high rate. 
In \cite{cheng2023linear}, the authors construct binary linear codes of rate $1/2-\varepsilon$ correcting $\Omega(\varepsilon^3 \log^{-1}(1/\varepsilon))$.
Can we ``binarize'' our construction and improve the fraction indels a high rate code can correct?
\end{itemize}

Finally, we would also like to mention a problem that was raised in \cite{cheng2023linear} which asks for the zero-rate threshold of a linear binary code. That is, what is the maximal $\delta$ for which for every $\varepsilon$, there exists a code with non-vanishing rate correcting $\delta - \varepsilon$. A trivial upper bound on the zero-rate threshold is $1/2$.
In \cite{guruswami2022zero}, the authors show that the zero-rate threshold for general binary indel codes is at most $1/2 - \delta_0$ (where $\delta_0$ is a tiny constant). 
For the case of linear binary indel codes, the authors of \cite{cheng2023linear} conjectured that even for codes with dimension $3$, there exists an absolute constant $\delta_0>0$ such that the code cannot correct $1/2-\delta_0$-fraction of deletions.

\balance
\bibliographystyle{IEEEtran}
\bibliography{bibliofile}
\end{document}